\documentclass[letterpaper,11pt,UKenglish]{article}

\usepackage[margin=1in]{geometry}
\usepackage[ttscale=.9]{libertine}
\usepackage[T1]{fontenc}
\usepackage{booktabs}

\usepackage{graphicx}
\usepackage{amsmath,amsthm,amssymb}
\usepackage{xspace}
\usepackage[colorlinks=true, allcolors=blue]{hyperref}
\usepackage{thm-restate}
\usepackage[table]{xcolor}
\usepackage{array}

\newcommand{\newbound}[1]{%
  \begingroup
  \setlength{\fboxsep}{2pt}%
  \colorbox{yellow!45}{\ensuremath{\displaystyle #1}}%
  \endgroup
}

\theoremstyle{plain}

\newtheorem{theorem}{Theorem}
\newtheorem{lemma}[theorem]{Lemma}
\newtheorem{corollary}[theorem]{Corollary}

\newtheorem*{claim*}{Claim}

\theoremstyle{definition}

\newcommand{\bigO}[1]{\mathcal{O}{\left(#1\right)}}
\newcommand{\bigOm}[1]{\Omega{\left(#1\right)}}
\newcommand{\bigTh}[1]{\Theta{\left(#1\right)}}

\newcommand{\etal}{\emph{et~al.}\xspace}
\newcommand{\eps}{\varepsilon}
\newcommand{\cD}{\mathcal{D}}
\newcommand{\cN}{\mathcal{N}}

\newcommand{\tD}{\widetilde{D}}
\newcommand{\Hyp}{\mathbb{H}}
\newcommand{\Euc}{\mathbb{E}}
\newcommand{\Reals}{\mathbb{R}}

\renewcommand{\angle}{\sphericalangle}
\newcommand\poly{\ensuremath{\mathrm{poly}}}

\DeclareMathOperator\arsinh{arsinh}

\newcommand{\inter}{\mathrm{int}}

\usepackage{todonotes}

\title{Fast Thick-Thin Decomposition for Sparse Spanners on Hyperbolic Surfaces}
\date{}
\author{Sándor Kisfaludi-Bak\thanks{Department of Computer Science, Aalto University, Espoo, Finland, \textsf{sandor.kisfaludi-bak@aalto.fi}.
This work was supported by the Research Council of Finland, Grant 363444.}
\and
Geert van Wordragen\thanks{Department of Computer Science, Aalto
    University, Espoo, Finland, \textsf{geert.vanwordragen@aalto.fi}}
}

\begin{document}

\maketitle

\begin{abstract}
    We consider spanners for point sets lying in the hyperbolic plane or on a closed hyperbolic surface with the restriction that spanner edges are not allowed to cross.
    This is a natural generalization of non-crossing Euclidean spanners. Thus, the resulting spanner graphs are embedded in the hyperbolic plane or on the hyperbolic surface.
    As our main contribution, we show that there are sparse $(1+\varepsilon)$-spanners for these problems when we are allowed to use Steiner points:
    \begin{itemize}
        \item on the hyperbolic plane we get a non-crossing Steiner $(1+\varepsilon)$-spanner with $\mathcal{O}(n / \varepsilon^2)$ edges,
        \item on hyperbolic surfaces of genus $g$ we get a Steiner $(1+\varepsilon)$-spanner
        \begin{itemize}
            \item with $\mathcal{O}(n / \varepsilon^{3/2} + g/\varepsilon^2)$ non-crossing edges, or
            \item with $\mathcal{O}(n / \sqrt{\varepsilon} + g/\varepsilon)$ edges that are allowed to cross.
        \end{itemize}
    \end{itemize}
    In particular, our spanners on surfaces have sparsity with linear dependence on $g$, rather than the easier-to-attain exponential dependence, and the terms $n/\varepsilon^{3/2}$ and $n/\sqrt{\varepsilon}$ match the current best Euclidean results for Steiner spanners with and without crossings.

    As a corollary of our non-crossing spanner and techniques from the existing literature on light spanners and minor-free TSP, we get an EPTAS for TSP on hyperbolic surfaces.

    Our surface constructions rely on the thick--thin decomposition, a standard tool for studying hyperbolic surfaces. For convex hyperbolic polygons, we introduce an analogous neck decomposition. We give algorithms that compute the thick--thin decomposition of a genus-$g$ surface in $\mathcal{O}(g^4\log g)$ time and the neck decomposition of an $n$-vertex polygon in $\mathcal{O}(n)$ time.

    
\end{abstract}
\thispagestyle{empty}
\clearpage

\setcounter{page}{1}
\section{Introduction}

Spanners are sparse graphs that preserve distances in a metric space up to a small multiplicative error.
In Euclidean and doubling metrics, such graphs are a central tool in geometric algorithms, and a rich theory is known even under additional topological restrictions such as planarity or non-crossing edges.
This paper asks to what extent this theory survives on hyperbolic surfaces: can one construct sparse Steiner $(1+\eps)$-spanners whose edges are embedded on the surface, and can this be done with dependence on the genus comparable to what one would hope for in bounded-genus graph algorithms?

An immediate obstacle is that hyperbolic surfaces combine finite area with potentially unbounded diameter.
A genus-$g$ closed hyperbolic surface has area $4\pi(g-1)$, but arbitrarily thin collars may force shortest paths to travel through long tube-like regions.
Thus local Euclidean spanner constructions cannot simply be patched together using a bounded-size global net.

The broader spanner literature provides a good starting point for what should be possible.
Classical constructions range from sparse weighted-graph spanners and Euclidean spanners to plane geometric spanners~\cite{Althofer93,Chew89,AryaDMS95}; spanners also underlie algorithmic distance structures such as well-separated pair decompositions and approximate distance oracles~\cite{CallahanK95,ThorupZ05}.
Sparse spanners have been developed extensively for metrics of bounded doubling dimension, using hierarchical nets and related decompositions~\cite{HarPeledM06}, with later work obtaining light spanners~\cite{Gottlieb15,BorradaileLW19}.
While the study of spanners and related sketching structures remains very active~\cite{LeS23,LeMS23,FiltserGN24,BhoreKKLLPT25,BodwinF25,BhoreCF26}, most of this work concerns doubling spaces, geometric graph classes, or graph metrics, with only a handful of papers considering non-doubling curved spaces~\cite{KrauthgamerL06,additivespanner,steinerspanner,Kisfaludi-BakW25} and only one that considers hyperbolic surfaces~\cite{Kisfaludi-BakW25}.

We next recall why the Steiner spanner setting is the right one for non-crossing $(1+\eps)$-spanners.
In the Euclidean plane, there has been a lot of work on \emph{plane spanners}, sometimes also called \emph{non-crossing} spanners, which are spanners whose edges do not cross, i.e., they form embedded planar graphs (henceforth called plane graphs).
The survey by Bose and Smid~\cite{planespannersurvey} covers much of what is known about plane spanners.
Plane spanners preserve the topology of the point set. They have $\bigO{n}$ edges and balanced separators of size $\bigO{\sqrt n}$~\cite{LiptonT79}, but it is a heavy restriction: there is a stretch lower bound of $1.4308$~\cite{planelowerbound} for Euclidean plane spanners.
In the hyperbolic plane, there is already a stretch lower bound of $2$ for spanners with a subquadratic number of edges, even if the edges are allowed to intersect~\cite{steinerspanner}.

To avoid these stretch lower bounds, one must add a set $S$ of \emph{Steiner points} to the initial point set $P$, and construct a geometric graph on $P\cup S$ such that, for every $p,q\in P$, the graph distance between $p$ and $q$ is at most $t$ times their geometric distance.
Clearly we can now make any (Steiner) spanner non-crossing by adding Steiner points at the edge intersections, but this will typically give a large edge count.
In particular, doing this to the complete graph yields a non-crossing Steiner $1$-spanner with $\bigO{n^4}$ Steiner points and edges in any planar metric.
With a more careful construction, one can make a plane Steiner $(1+\eps)$-spanner with $\bigO{n / \eps^2}$ edges for $L_1$ distances and a $(1+\eps)$-spanner for Euclidean distances with $\bigO{n / \eps^4}$ edges~\cite{planesteiner} (also explained in~\cite{planespannersurvey, zeh2002efficient}).
When all angles in the Delaunay triangulation are at least~$\alpha$, Borradaile and Eppstein~\cite{wellspaced} give a Euclidean plane Steiner $(1+\eps)$-spanner that only has $\bigO{\frac{n}{\alpha^2 \eps^3} \log^2 \frac{1}{\alpha\eps}}$ edges (while also focusing on \emph{lightness}).
Recently, these results were improved to $\bigO{n / \eps^{3/2}}$ edges and this was shown to be almost tight:
there is a lower bound $\bigOm{n / \eps^{3/2 - \mu}}$ for any fixed $\mu > 0$~\cite{planesteinernew}.
Thus the Euclidean plane gives the natural target bounds: $\bigO{n/\sqrt{\eps}}$ for Steiner spanners with crossings and $\bigO{n/\eps^{3/2}}$ for non-crossing Steiner spanners.
Our goal is to obtain linear-size spanners with similar dependencies on $\eps$, and small (preferably polynomial) dependence on the genus.

We now turn to the surface setting.
By the uniformisation theorem, all closed orientable surfaces other than the sphere and the torus can be endowed with a hyperbolic metric, turning them into hyperbolic surfaces.
This makes them an important object of study in (computational) topology.
Hyperbolic metrics provide a standard geometric model for surfaces of negative Euler characteristic, allowing one to think of topological objects such as curves, decompositions, and cut systems as metric objects built from geodesics, collars, and injectivity radii~\cite{thurston97three,Buser10}.
This geometric viewpoint is also central in computational topology, where many algorithmic questions concern graphs, curves, homotopy, and optimization problems on surfaces~\cite{Colin17}.


\begin{figure}
    \centering
    \includegraphics{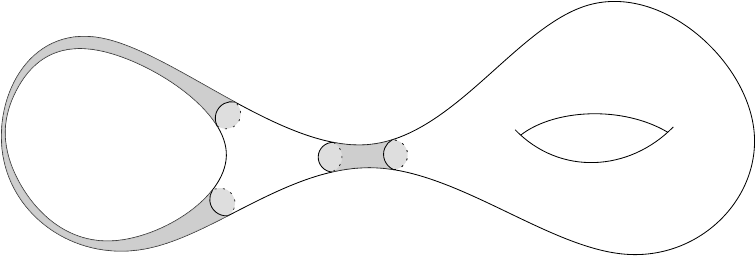}
    \caption{Illustration of the thick-thin decomposition of a genus-$2$ surface into two thin parts and two thick parts. The two thin parts (grey) are homeomorphic to cylinders.}
    \label{fig:thickthin}
\end{figure}

The standard geometric tool for handling the thin-collar phenomenon is the \emph{thick-thin decomposition}.
Due to the Gauss-Bonnet theorem, a hyperbolic surface of genus $g$ has area $(g-1) \cdot 4\pi$, but parts of it can be extremely thin, leaving us without a bound on the diameter.
The thick-thin decomposition splits the surface into thin tube-like parts and thick parts where we can now (for example) bound the diameter. See Figure~\ref{fig:thickthin} for an illustration.
For spanners, this suggests a natural strategy: use local Euclidean-like constructions in bounded-diameter thick regions, and handle the long thin collars separately.
The difficulty is to turn this geometric intuition into an efficient algorithm.

The size of an $\eps$-net for a hyperbolic surface is not bounded by any function of $g$ and $\eps$ alone~\cite{epsnet}, while when we restrict to the thick parts of the surface its size is $\bigO{g/\eps^2}$.
Lanuel~\cite{lanuelthesis} gives an algorithm to compute $\eps$-nets of the thick parts (called a \emph{pseudo-$\eps$-net}) and uses them to obtain the thick--thin decomposition, but does not analyse its running time dependence on surface-specific values such as $g$.
Based on our analysis, both their pseudo $\eps$-net algorithm and the algorithm to cover the entire surface~\cite[Theorem~III.4]{lanuelthesis} have an exponential dependence on $g$.
Notably, this algorithm is practical for small $g$ and has been successfully implemented~\cite{DespreLPT25}.
The thick-thin decomposition has been used in~\cite{Despre26} in an algorithmic fashion, while~\cite{EbbensPV23} uses its existence in its algorithm design.
To date, all such algorithmic approaches have exponential dependence on the genus $g$, due to the fact that in order to cover shortest paths on the surface, one needs to look in the covering space $\Hyp^2$.
Despr\'e, Kolbe, and Teillaud show that shortest paths on the surface can intersect up to $\bigO{g}$ copies of the Dirichlet domain~\cite{DespreKT24}.
Since a Dirichlet domain has $\bigO{g}$ sides, this leads naturally to exploring $g^{\bigO{g}}$ domain copies in algorithms that must cover all such lifted shortest paths or edges.
A central algorithmic challenge is therefore to avoid this universal-cover explosion while still finding enough of the surface geometry to build sparse spanners.

The state of the art for Steiner spanners on hyperbolic surfaces leaves exactly this gap.
To the best of our knowledge, no previous work studies non-crossing (Steiner) spanners on a curved surface.
When crossings are allowed, there are Steiner-spanner constructions in the hyperbolic plane by Krauthgamer and Lee~\cite{KrauthgamerL06}, by Kisfaludi-Bak and Van Wordragen~\cite{steinerspanner,Kisfaludi-BakW25}, and by Park and Vigneron~\cite{additivespanner}.
For hyperbolic surfaces, the only previous spanner result is by Kisfaludi-Bak and Van Wordragen~\cite{Kisfaludi-BakW25}, which has $\frac{n}{\sqrt\eps} \log\frac1\eps \cdot g^{\bigO g}$ edges.\footnote{An edge count of $\bigO{\frac{n}{\sqrt\eps} \log\frac1\eps}$ was attained in~\cite{Kisfaludi-BakW25} for the sphere and flat tori, taking care of the case of non-hyperbolic constant-curvature orientable surfaces.}
None of these lead to results for non-crossing (Steiner) spanners.
Table~\ref{tab:steiner-spanner-sparsity} summarizes the relevant sparsity bounds.

\begin{table}[t]
\centering
\footnotesize
\setlength{\tabcolsep}{4pt}
\renewcommand{\arraystretch}{1.25}
\begin{tabular}{@{}
  >{\centering\arraybackslash}m{0.13\linewidth}
  >{\centering\arraybackslash}m{0.25\linewidth}
  >{\centering\arraybackslash}m{0.25\linewidth}
  >{\centering\arraybackslash}m{0.29\linewidth}
@{}}
\toprule
\textbf{Bound} &
\textbf{Euclidean plane} &
\textbf{Hyperbolic plane} &
\shortstack{\textbf{Hyperbolic surface}\\\textbf{of genus $g$}} \\
\midrule

\multicolumn{4}{@{}c@{}}{\small\bfseries Steiner $(1+\eps)$-spanners, crossings allowed} \\
\cmidrule(lr){1-4}
\addlinespace[2pt]

Upper
& $\bigO{n/\sqrt{\eps}}$~\cite{planesteinernew,LeS25,ChangC0MST24}
& $\bigO{n\log(1/\eps)/\sqrt{\eps}}$~\cite{Kisfaludi-BakW25}
& \begin{tabular}[c]{@{}c@{}}
  $\bigO{n g^{\bigO{g}}\log(1/\eps)/\sqrt{\eps}}$~\cite{Kisfaludi-BakW25} \\
  $\widetilde{\mathcal{O}}(ng/\eps)$, existential~\cite{BhoreKK0LPT25} \\
  \newbound{\bigO{n/\sqrt{\eps}+g/\eps}}~Thm.~\ref{thm:main}
  \end{tabular}
\\

Lower
& $\bigOm{n/\sqrt{\eps}}$~\cite{BhoreT22}
& $\bigOm{n/\sqrt{\eps}}$~\cite{BhoreT22}
& $\bigOm{n/\sqrt{\eps}}$~\cite{BhoreT22}
\\

\midrule

\multicolumn{4}{@{}c@{}}{\small\bfseries Non-crossing Steiner $(1+\eps)$-spanners} \\
\cmidrule(lr){1-4}
\addlinespace[2pt]

Upper
& $\bigO{n/\eps^{3/2}}$~\cite{planesteinernew}
& \begin{tabular}[c]{@{}c@{}}
  no previous bound \\
  \newbound{\bigO{n/\eps^2}}~Thm.~\ref{thm:SteinerPlaneSpanner}
  \end{tabular}
& \begin{tabular}[c]{@{}c@{}}
  no previous bound \\
  \newbound{\bigO{n/\eps^{3/2}+g/\eps^2}}~Thm.~\ref{thm:main}
  \end{tabular}
\\

Lower
& $\bigOm{n/\eps^{3/2-\mu}}$~\cite{planesteinernew}
& $\bigOm{n/\eps^{3/2-\mu}}$~\cite{planesteinernew}
& $\bigOm{n/\eps^{3/2-\mu}}$~\cite{planesteinernew}
\\

\bottomrule
\end{tabular}
\caption{Edge counts for Steiner $(1+\eps)$-spanners. The lower bounds involving $\mu$ hold for every fixed $\mu>0$. The lower bounds in the hyperbolic plane and on hyperbolic surfaces follow from the Euclidean lower-bound constructions placed in a sufficiently small embedded disk. The existential polyhedral-surface upper bound transfers to
hyperbolic surfaces by approximation, but this transfer is non-algorithmic.}
\label{tab:steiner-spanner-sparsity}
\end{table}

Finally, we mention one related point of comparison regarding the size of our non-crossing spanners. In a different line of work, Kapoor and Li~\cite{KapoorL09} and Bhore \etal~\cite{BhoreKK0LPT25} considered polyhedral surfaces. One can approximate hyperbolic surfaces with polyhedral surfaces by first finding a triangulation of the surface with triangles of tiny hyperbolic diameter $\delta>0$, then creating a corresponding polyhedral surface with Euclidean triangles of the same side lengths; as $\delta$ goes to $0$, so does the distortion between a Euclidean and hyperbolic triangle of the same side lengths.
Thus the existence of spanners on polyhedral surfaces also implies spanners of the same edge count and lightness in the hyperbolic setting, but we emphasize that algorithms to construct spanners do not transfer easily from (Euclidean) polyhedral surfaces to hyperbolic surfaces.

\subsection{Our contribution}
Our first contribution is an algorithm to compute thick-thin decompositions in polynomial time rather than in exponential time. We believe that this is interesting in its own right and will be used in future algorithms on hyperbolic surfaces. The algorithm requires that the surface is defined via a Dirichlet domain, which is a convex polygon with $\bigO g$ vertices in the hyperbolic plane. This is a reasonable input model, as one can compute a Dirichlet domain from any given fundamental polygon \cite{DespreKPT25}. See Section~\ref{sec:prelim} for more details on our input model.

The thick-thin decomposition is based on the \emph{injectivity radius}:
for a given point $p$, this value $r_p$ is the supremum of the values $r$ for which the set of points within distance $r$ from~$p$ forms a topological disk.
Equivalently, $2r_p$ is the length of the shortest non-contractible geodesic loop through~$p$.
The \emph{$\mu$-thick-thin decomposition} now assigns points from a hyperbolic surface to either the thick part or the thin part depending on whether that point's injectivity radius is above or below a given value $\mu$.
If $\mu$ is at most the Margulis constant $\mu_M = 2\arsinh\sqrt{\frac{2\cos(2\pi/7) - 1}{8\cos(\pi/7) + 7}} \approx 0.2629$ of the hyperbolic plane~\cite{margulis}, then each connected thin component is homeomorphic (and in fact diffeomorphic) to a cylinder \cite[Theorem 4.5.6]{thurston97three}; see Figure~\ref{fig:thickthin}.
Algorithmically, we consider the thick-thin decomposition of a given fundamental domain polygon $\tD$ to be a partition of the domain into thick and thin regions. We show the following theorem.

\begin{restatable}{theorem}{thmthickthin}\label{thm:thickthin}
    For any value $\mu \leq \mu_M$, we can construct the $\mu$-thick-thin decomposition of a hyperbolic surface of genus $g$ given by any Dirichlet domain $\tD$ in time $\bigO{g^4 \log g}$.
\end{restatable}

The polynomial dependence on $g$ is surprising; intuitively, it means that the algorithm must avoid looking at the $g^{\bigO g}$ copies of the domain required to cover all shortest paths on the surface.

As our main contribution, we use our thick-thin decomposition to construct sparse Steiner spanners for points on hyperbolic surfaces, with and without crossings.

\begin{restatable}[Main theorem]{theorem}{thmsurfspanner}\label{thm:main}
    Let $S$ be a closed hyperbolic surface of genus $g$.
    For any set $P$ of $n$ points on~$S$ and $\eps\in (0,1/2)$, there is a non-crossing Steiner $(1+\eps)$-spanner with $m = \bigO{n / \eps^{3/2} + g/\eps^2}$ edges and a Steiner $(1+\eps)$-spanner (with crossings) with $m = \bigO{n / \sqrt{\eps} + g/\eps}$ edges.
    Given $S$ as a Dirichlet domain, we can construct these in $\bigO{(g^4 + g^2 \log\frac1\eps) \log(g\log\frac1\eps) + m \log n}$ time.
\end{restatable}

Note that for $g = \bigO{n\sqrt\eps}$ both edge counts match their Euclidean counterparts~\cite{planesteinernew,LeS25,ChangC0MST24}, which are optimal in their dependence on $n$ and, up to $(1/\eps)^{o(1)}$ factors, in their dependence on $\eps$~\cite{planesteinernew,BhoreT22}.
Additionally, the dependence on $g$ is only linear and not multiplied with~$n$, which improves on the edge count $\frac{1}{\sqrt\eps} \log\frac1\eps \cdot n g^{\bigO g}$ found earlier for Steiner spanners with crossings~\cite{Kisfaludi-BakW25}.
For non-crossing Steiner spanners, this is the first construction in any hyperbolic setting.

Using very similar techniques we also get results in the hyperbolic plane.
Namely, we introduce the \emph{neck decomposition} as an analogue of the thick--thin decomposition for hyperbolic polygons and give a linear-time algorithm to compute this using essentially the Dobkin-Kirkpatrick Hierarchy~\cite{DobkinKirkpatrick}.
The neck decomposition has properties similar to those of the thick--thin decomposition, giving the following result.

\begin{restatable}{theorem}{thmplanespanner}\label{thm:SteinerPlaneSpanner}
    For any point set $P \subset \Hyp^2$ of $n$ points and any $\eps\in (0,1/2)$, we can construct a non-crossing Steiner $(1+\eps)$-spanner with $\bigO{n / \eps^2}$ edges in $\bigO{n \log n / \eps^2}$ time.
\end{restatable}

In all of our algorithms, we use the real RAM model of computation; for a precise definition of the model, see~\cite{EHM2022Smoothing}.

\paragraph*{Light non-crossing spanners and a corollary for TSP.} Note that, because non-crossing (Steiner) spanners are $K_r$-minor-free graphs (for $r = 5$ in the hyperbolic plane and some $r = \bigO{\sqrt g}$ on surfaces~\cite{RingelYoung}), we can use Theorem~1.5 of Le and Solomon~\cite{LeS23} to also get bounded \emph{lightness}.
If we define the \emph{weight} of a graph as the sum of its edge weights, then lightness is the ratio between the weight of the spanner and that of the minimum spanning tree for the same point set.
We now specifically get lightness $\widetilde{\mathcal O}(\frac{1}{\eps^2})$ in the hyperbolic plane and $\widetilde{\mathcal O}(\frac{\sqrt g}{\eps} + \frac{1}{\eps^2})$ on hyperbolic surfaces, where $\widetilde{\mathcal O}$ hides a polylogarithmic factor in $1/\eps$ and $g$.
This gives a partial answer to the open question of finding light spanners in hyperbolic space~\cite{steinerspanner,hypertsp} and removes one obstacle preventing better running times for the travelling salesman problem (TSP) in the hyperbolic plane. In TSP, the goal is to find the shortest closed curve visiting the given set of points in some geometric space. Using a recent subset TSP algorithm on edge-weighted minor-free graphs~\cite{CohenAddadFKL20} within our spanner (and the input point set as terminals), we obtain the following corollary.

\begin{corollary}
    Given a set of $n$ points on a hyperbolic surface of genus $g$ (which in turn is given as a Dirichlet domain), we can compute in time $2^{\mathcal{O}_g(\poly(1/\eps))}\poly(n)$ a TSP tour of the points of length at most $1+\eps$ times the shortest tour.
\end{corollary}

\subsection{Brief overview of ideas and techniques}

\paragraph{Thick-thin decomposition.} Consider a hyperbolic surface $\Hyp^2/\Gamma$ given as a Dirichlet domain $\tD$, which is a convex polygon in the hyperbolic plane with $\bigO g$ edges. The hyperbolic polygons $\gamma\tD$ for $\gamma\in \Gamma$ will be referred to as the \emph{copies of $\tD$}. We observe that a $2\mu_M$-neighbourhood of $\tD$ will cover the lifts of all closed geodesic loops of length $2\mu_M$ with one endpoint in $\tD$, that is, any geodesic loop that is a core curve of a thin part collar can be discovered by inspecting copies of $\tD$ that intersect this $2\mu_M$-neighbourhood. Once such a short geodesic loop is found, one can identify the corresponding thin part and compute the regions in $\tD$ that belong to it.

The idea therefore is to grow the parameter $\delta$ from $0$ to $\mu_M$, and find the copies $\gamma\tD$ (for $\gamma\in \Gamma$) intersecting the $\delta$-neighbourhood of $\tD$. To get a polynomial running time, we need to bound the number of such copies intersecting this neighbourhood by a polynomial in $g$. Unfortunately, the number of such copies can be unbounded, as a constant-neighbourhood of a thin part may intersect an unbounded number of copies of $\tD$. Luckily, there are only a small number of domain copies near the thick parts.

\begin{restatable}{lemma}{thicktersect}\label{lem:thicktersect}
    Let $\widetilde{T} \subseteq \tD$  denote the lifting of the union of the $\mu_M$-thick parts of $S$. 
    Then, $N_{\Hyp^2}(\widetilde{T}, 2\mu_M)$ intersects $\bigO{g^3}$ copies of $\tD$.
\end{restatable}

To tackle the problem near thin parts, we can stop exploring their neighbourhood right after the thin parts are detected. The following lemma helps us bound the number of copies of $\tD$ we need to explore near a thin part. Recall that $r_p$ denotes the injectivity radius of the point $p$.

\begin{restatable}{lemma}{thintersect}\label{lem:thintersect}
    Let $\widetilde C \subseteq \tD$ denote the lifting of a $\mu_M$-thin part of $S$ and let $\widetilde N$ denote the region covered by balls of radius $r_p$ around each point $p \in \widetilde C$.
    Then, $\widetilde N$
    intersects $\bigO{g}$ copies of $\tD$.
\end{restatable}

With these key lemmas at hand and a Euclidean dynamic point location data structure~\cite{ChanN18}, the algorithm for the thick-thin decomposition follows.

\paragraph{Crossing and non-crossing spanners on surfaces}
The proof of our main theorem builds heavily on a recent construction for the same problem in the Euclidean plane~\cite{planesteinernew}, where they achieve a plane Steiner spanner of stretch $1+\eps$ using $\bigO{n/\eps^{3/2}}$ edges, and note that the same idea gives a non-plane Steiner spanner with $\bigO{n/\sqrt\eps}$ edges. We can use these constructions in a small-diameter neighbourhood, which is unsurprising, but in fact far from simple.
First, a brute-force embedding of a constant-radius disk gives distortion proportional to the diameter of said disk, leading to constant multiplicative error. Thus, one needs to adapt the construction of~\cite{planesteinernew} more carefully. 

Unfortunately, the construction of~\cite{planesteinernew} is based on angles, cones, straight lines, and axis-parallel balanced box decompositions, none of which work in the hyperbolic setting. There are no balanced box decompositions, and no model of the hyperbolic plane has straight geodesics \emph{and} conformal angles at the same time, i.e., either we have curved geodesics and conformal angles, or straight geodesics and distorted angles. Here, we opt for straight geodesics and work in the Beltrami-Klein model, where we apply the Euclidean construction of~\cite{planesteinernew} in a more or less black-box fashion.  After proving the stretch bound within one small neighbourhood, we need to consider neighbourhoods and rotations with different centres. It is non-trivial to prove that the number of crossings between different rotations is still small. Recall that for non-crossing spanners, all crossings must be replaced with a Steiner point, directly impacting the edge count. In the end, we successfully prove the following theorem.

\begin{restatable}{theorem}{thmconstdiam}\label{thm:constdiam}
    For any set $P \subset \Hyp^2$ of $n$ points with diameter at most $\frac{1}{4}$ and any $\eps > 0$, we can construct a Steiner $(1+\eps)$-spanner (with crossings) with $\bigO{n / \sqrt{\eps}}$ edges in $\bigO{n \log n / \sqrt{\eps}}$ time, and a non-crossing Steiner $(1+\eps)$-spanner with $\bigO{n / \eps^{3/2}}$ edges in $\bigO{n \log n / \eps^{3/2}}$ time.
    Moreover, if we have two point sets $P_1,P_2 \subset \Hyp^2$ of $n$ points in total, then their non-crossing Steiner $(1+\eps)$-spanners constructed in this manner will only intersect in $\bigO{n / \eps^{3/2}}$ points.
\end{restatable}

To build our spanners, we first build a pseudo-$\frac{1}{16}$-net $C$ (i.e., a $\frac{1}{16}$-net of the thick parts), which is easily constructed once a thick-thin decomposition is available. For each point in $C$ we consider a circle of constant radius and add Steiner points at distance $\sqrt{\eps}$ along the circle; let $S$ denote the Steiner points added on these circles. (The construction in the thin parts is slightly more technical and not discussed here.) We apply the constant-diameter spanner construction on the points from the input and $S$ in each of these circles.


To see why the resulting spanner has stretch $1+\eps$ in each thick part, one can see that any point pair within some small constant distance will be covered by some disk in which we have applied the constant-diameter construction. For point pairs at greater distances, we observe that along the shortest path connecting the pair of points on the surface, we can find a point $s_i\in S$ at distance $\bigO{\sqrt{\eps}}$ from the shortest path, and at distance $\Theta(1)$ from the previous point $s_{i-1}\in S$ near the shortest path. Thus, by a lemma of~\cite{Kisfaludi-BakW25}, the spanner path between $s_{i-1}$ and $s_i$ is a $(1+\bigO{\eps})$-distortion of the corresponding constant-length section of the shortest path. The argument for traversing the thin parts is slightly more involved.

Finally, the algorithm to construct the spanner requires some supporting data structures, including hyperbolic approximate nearest neighbours~\cite{Kisfaludi-BakW25}.

\paragraph{Non-crossing spanner on the hyperbolic plane}

Unfortunately our hyperbolic surface result for genus $g$ does not extend to the hyperbolic plane simply with $g=0$, since the surface area bound $\bigO g$ is crucially used in our construction. In the hyperbolic plane, the area is infinite, but we can consider the convex hull of the input point set; this has area at most $(n-2)\pi$. We introduce so-called \emph{neck decompositions} for convex hyperbolic polygons: they consist of thin parts called \emph{necks} (where two edges of the polygon run within constant distance of each other for an extended length) and the remaining thick parts. One can think of thick parts as branching regions/hubs that are attached into a tree-like structure, where necks act as edges. 
We can compute this decomposition in linear time using a recursion that resembles the Dobkin-Kirkpatrick Hierarchy~\cite{DobkinKirkpatrick}. With the decomposition at hand, the construction is very much analogous to our construction on surfaces.


\section{Preliminaries}\label{sec:prelim}

\paragraph{Hyperbolic geometry.}
To work with hyperbolic surfaces, we will frequently need to use the hyperbolic plane $\Hyp^2$.
We give the relevant details here, but refer to \cite{cannon1997hyperbolic} and the textbooks \cite{iversen1992hyperbolic,thurston97three,benedetti1992lectures} for a deeper understanding.
Though most of our proofs do not use any particular model of $\Hyp^2$ (and our figures are often not drawn in any particular model either), our algorithms ultimately work in the (Poincaré) half-plane model.
Here, any point in $\Hyp^2$ is represented by a complex number with positive imaginary part.
The distance between points $a,b$ differs significantly from their Euclidean distance $|a-b|$ and is given by $2\arsinh\sqrt{\frac{|a-b|^2}{4\Im(a)\Im(b)}}$.
Note that we cannot compute $\arsinh$ and square roots in the real-RAM model, but since they are strictly increasing we can nevertheless compare distances to each other and to constants of the form $2\arsinh\sqrt{c}$.
Thus, whenever we mention comparing distances with a specific given constant we will in fact compare them with an appropriate approximation of the constant in this given form.
In the half-plane model, all important curves appear as either Euclidean line segments or circle arcs.
In particular, this includes when the curve is one of the following:
a hyperbolic line segment, the shortest curve $pq$ connecting $p$ to $q$;
a circle, the set of points at equal distance from a centre;
or a hypercycle, the set of points at equal distance from a (hyperbolic) line.

We do not use any other particular properties of the half-plane model, other than that it can accurately represent particular tilings.
Namely, from \cite{steinerspanner} we can extract the following tilings by taking an appropriate level of the quadtree and replacing every curve (horocycle) forming the top boundary of some cell with a line segment.
\begin{theorem}[Tilings from \cite{steinerspanner}]\label{thm:quadtiling}
    For any value $r \geq 0$ there is a tiling $Q_r$ of $\Hyp^2$ where each tile is an irregular $2^{\bigO{1+r}}$-gon of diameter at most $r$ that contains a ball of radius $\bigOm{r}$, and each vertex has degree at most four.
\end{theorem}

\paragraph{Hyperbolic surfaces.}
To be able to work with surfaces algorithmically, it is important to define how they appear in our input. For a thorough introduction to hyperbolic surfaces, see the books~\cite{Beardon83,Massey91}. It is well-known that all hyperbolic surfaces are locally isometric to the hyperbolic plane $\Hyp^2$. They can be obtained as the quotient of $\Hyp^2$ under the action of $\Gamma$, which is some discrete subgroup of isometries of~$\Hyp^2$ (analogous to how a square torus can be seen as a quotient of the Euclidean plane). A \emph{Dirichlet domain} $\tD_{\tilde x}$ of a point $\tilde x\in \Hyp^2$ is the (closed) Voronoi cell of $\tilde x$ in the Voronoi diagram of $\Gamma \tilde x$. (The Voronoi diagram is the decomposition of $\Hyp^2$ into convex polygonal cells according to the nearest point of $\Gamma \tilde x$). It is well-known that $\tD_{\tilde x}$ is a \emph{fundamental domain} of the surface $S=\Hyp^2/\Gamma$, that is, $\Gamma \tD_{\tilde x}=\Hyp^2$ and for each point $y\in \inter \tD_{\tilde x}$ we have $\gamma y \neq \gamma'y$ for distinct $\gamma,\gamma'\in \Gamma$.
Let $\rho : \Hyp^2 \to S$ denote the projection map.
For a point $p \in S$, we use $\tilde{p} \in \Hyp^2$ to denote a lift, i.e.\ $\rho(\tilde p) = p$ or equivalently $\tilde p$ is a representative in $\Hyp^2$ of the element $p \in \Hyp^2 / \Gamma$.
We also extend this to sets of points.

\paragraph{Input model.} In our algorithms, the Dirichlet domain is represented by the cyclic sequence of its vertices in the Poincaré half-plane model, its centre, and its side-pairing information. The corresponding Möbius transformations are stored as projective $2\times2$ real matrices; we do not normalize their determinants to one. From the side-pairing data, these matrices can be computed using only standard arithmetic operations in $\bigO g$ time. The points on the surface must fall in the given Dirichlet domain.


\section{Thick-thin decomposition for hyperbolic surfaces}

First, we give an algorithm to compute the thick-thin decomposition of a genus $g$ hyperbolic surface $S = \Hyp^2 / \Gamma$, given by a Dirichlet domain $\tD$.

\begin{figure}
    \centering
    \includegraphics{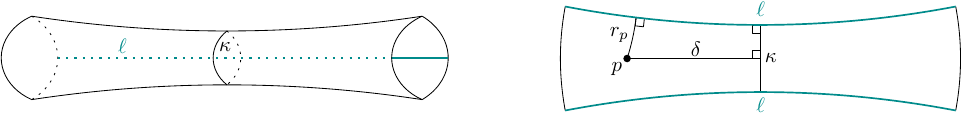}
    \caption{A thin part of the surface (left) can be cut open along a curve $\ell$ to get a shape in $\Hyp^2$ bounded by two copies of $\ell$ and two hypercycles (right).
    The Lambert quadrilateral on the right now lets us calculate the injectivity radius $r_p$ for a point $p$ in the thin part.}
    \label{fig:thinlambert}
\end{figure}

Thin parts have a lot of structure, which is what makes the thick-thin decomposition useful.
As mentioned, each thin part $T$ is homeomorphic to a cylinder, but on top of that, each can also be defined as the neighbourhood of some (short) closed geodesic $\kappa$ \cite[Theorem 4.5.6]{thurston97three}.
Consequently, we can cut $T$ open with a geodesic $\ell$ perpendicular to $\kappa$ to get a region in $\Hyp^2$ bounded by two copies of $\ell$ and two hypercycles (see Figure~\ref{fig:thinlambert}).
The length of $\kappa$ completely determines the shape of~$T$.
In particular, for a point $p$ at distance $\delta$ from $\kappa$, we get the relation $\sinh r_p = \sinh\frac{|\kappa|}{2} \cdot \cosh\delta$ from the \emph{Lambert quadrilateral} shown in Figure~\ref{fig:thinlambert}; a quadrilateral with three right angles and one acute angle at $p$.
This also fixes the distance $\delta_{\max}$ of $\kappa$ to the boundary of $T$ as $\cosh \delta_{\max} = \sinh \mu / \sinh\frac{|\kappa|}{2}$.

\begin{lemma}\label{lem:thinjradius}
    A point $p$ in a $\mu_M$-thin part at distance $\delta$ from the boundary has injectivity radius $r_p = \bigOm{e^{-\delta}}$.
    If $\delta \geq 1$, then $r_p = \bigTh{e^{-\delta}}$.
\end{lemma}
\begin{proof}
    From the equations above and using the difference of arguments identity for $\cosh$,
    \begin{align*}
        \sinh r_p
        &= \sinh\frac{|\kappa|}{2} \cdot \cosh(\delta_{\max} - \delta) \\
        &= \sinh\frac{|\kappa|}{2} \cdot (\cosh\delta_{\max}\cosh\delta - \sinh\delta_{\max}\sinh\delta) \\
        &= \sinh\frac{|\kappa|}{2} \cosh\delta_{\max} \cdot (\cosh\delta - \tanh\delta_{\max} \sinh\delta) \\
        &= \sinh\mu_M \cdot (\cosh\delta - \tanh\delta_{\max} \sinh\delta) \\
        &= \bigTh{\cosh\delta - \tanh\delta_{\max} \sinh\delta}.
    \end{align*}
    For the lower bound, first note that $\tanh x \leq 1$ and $\sinh x \geq 0$ for $x \geq 0$, so
    \begin{align*}
        \cosh\delta - \tanh\delta_{\max} \sinh\delta
        \geq \cosh\delta - \sinh\delta
        = e^{-\delta}.
    \end{align*}
    For the upper bound, we use $\tanh\delta_{\max} \geq \tanh \delta$ to obtain
    \begin{align*}
        \cosh\delta - \tanh\delta_{\max} \sinh\delta
        &\leq \cosh\delta - \tanh\delta \sinh\delta \\
        &= \cosh\delta - \sinh^2\delta / \cosh\delta \\
        &= \cosh\delta - (\cosh^2\delta - 1) / \cosh\delta \\
        &= 1 / \cosh\delta.
    \end{align*}
    Finally we use that $\cosh x = \bigTh{e^x}$ for $x \geq 1$, so that for $\delta \geq 1$ we also have an upper bound $\bigO{e^{-\delta}}$.
    Now, applying that $\arsinh x = \bigTh{x}$ for $x \leq 1$ gives the claimed results.
\end{proof}

For the proofs that follow we will need these two properties of Dirichlet domains from \cite{DespreKT24}.

\begin{lemma}[Proposition 4 of \cite{DespreKT24}]\label{lem:dirichletboundary}
    Any Dirichlet domain for a hyperbolic surface of genus~$g$ has a boundary made up of at most $24g - 12$ shortest curves.
\end{lemma}

\begin{lemma}[Lemma 5 of \cite{DespreKT24}]\label{lem:dirichletcurve}
    Let $\tD$ be a Dirichlet domain with $k$ edges.
    A shortest curve between two points crosses the boundary of $\tD$ at most $k/2$ times.
\end{lemma}

Let $N_X(R, \delta)$ denote the $\delta$-neighbourhood of $R$, i.e.\ the points in $X$ with distance at most $\delta$ from a point in $R$.
We only need to consider specific neighbourhoods of the thick and thin parts, which lets us bound the number of domains we actually need to look at; see Figure~\ref{fig:dirichletneighbourhood}.

\begin{figure}
    \centering
    \includegraphics{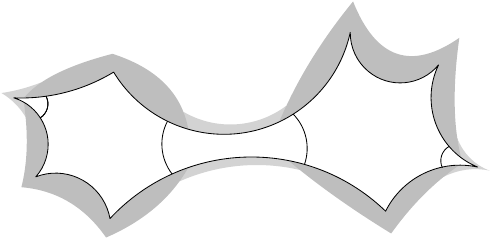}
    \caption{To find the thick--thin decomposition, we need only consider the dark grey neighbourhood of thick parts and the light grey neighbourhood of thin parts.}
    \label{fig:dirichletneighbourhood}
\end{figure}

\thicktersect*

\begin{proof}
    First, we show that the boundary of $\widetilde{T}$ has length $\bigO{g^2}$, which lets us say that $N_{\Hyp^2}(\widetilde{T}, 2\mu_M)$ has area $\bigO{g^2}$.
    For this, consider a shortest curve $\kappa$ of length $l$ that lies in a single $\mu_M$-thick part.
    Then, $N_S(\kappa, \mu_M)$ is a topological disk, meaning it has the area we would expect it to have in $\Hyp^2$, which is $2l \sinh \mu_M + 2\pi\cosh(\mu_M-1)= \bigTh{l+1}$.
    Since $S$ has total area $\bigO{g}$, this now implies $l = \bigO{g}$.    
    By Lemma~\ref{lem:dirichletboundary}, the boundary of $\tD$ consists of $\bigO{g}$ shortest curves.
    Since there are $\bigO{g}$ thin parts, we can split up each curve into $\bigO{g}$ ones that are each contained in some thick part and $\bigO{g}$ ones that are contained in some thin part.
    From this it follows that the part of the boundary of $\widetilde{T}$ it shares with $\tD$ has length $\bigO{g^2}$.
    The boundary of $\widetilde{T}$ can also come from the boundary between a thick and a thin part.
    Since there are $\bigO{g}$ thin parts, each is adjacent to at most two thick parts, and this shared boundary has length $\bigO{\mu_M}$, this contributes length $\bigO{g}$.
    Thus, the boundary of $\widetilde{T}$ has length $\bigO{g^2}$ and $\bigO{g}$ vertices. Consequently, $N_{\Hyp^2}(\widetilde{T}, 2\mu_M)$ has boundary length at most $\bigO{g} + \bigO{g^2} \cdot \cosh(2\mu_M) = \bigO{g^2}$, which, by the hyperbolic isoperimetric inequality~\cite{Osserman78}, also bounds its area.

    Now, let $T' = N_S(\rho(\widetilde T), 2\mu_M)$ and observe that for any $p \in T'$ we have $r_p \geq c$ for some constant $c$, due to Lemma~\ref{lem:thinjradius}.
    Let $Q_c$ be as in Theorem~\ref{thm:quadtiling}.
    Now, consider the intersection of $Q_c$ with $N_{\Hyp^2}(\widetilde{T}, 2\mu_M)$.
    By construction the tiles of $Q_c$ intersecting $N_{\Hyp^2}(\widetilde{T}, 2\mu_M)$ are contained in $N_{\Hyp^2}(\widetilde{T}, 3\mu_M)$ which still has area $\bigO{g^2}$, and seeing as the tiles have area $\bigTh{1}$ we must have $\bigO{g^2}$ tiles.
    The tile boundaries contained in $N_{\Hyp^2}(\widetilde{T}, 2\mu_M)$ have length at most $\mu_M$ and are thereby shortest curves on $S$, which means by Lemma~\ref{lem:dirichletcurve} that they each intersect $\bigO{g}$ elements of $\Gamma D$.
    Thus, $N_{\Hyp^2}(\widetilde{T}, 2\mu_M)$ intersects $\bigO{g^3}$ elements of $\Gamma \tD$ in total.
\end{proof}

We give a similar result for thin parts, except now we use a finer definition of neighbourhood.

\thintersect*
\begin{proof}
    We can cover $\widetilde N$ with two additional copies of $\widetilde C$ and a $\mu_M$-neighbourhood of the thin part's boundary.
    Let $p,p',q,q'$ be points along the boundary of the thin part $C$ on the surface at maximum distance from each other.
    The four triangles formed from the shortest curves between these points cover most of $C$, but since the boundary of $C$ is formed by hypercycles of constant length there is a region $R$ of constant diameter left uncovered.
    We take the union of $R$ with the $\mu_M$-neighbourhood of the thin part's boundary to get a new region $R'$ which still has constant diameter.
    As in the proof of Lemma~\ref{lem:thicktersect}, we note that all points in $R'$ are close to the boundary and thus Lemma~\ref{lem:thinjradius} implies their injectivity radius is larger than some constant $c$.
    Let $Q_c$ be as in Theorem~\ref{thm:quadtiling}; $R'$ is intersected by a constant number of these tiles and each is bounded by a constant number of shortest curves.
    Thus, we can cover $\widetilde N$ with a constant number of triangles whose edges are shortest curves.
    In turn, each triangle intersects $\bigO{g}$ copies of the Dirichlet domain by Lemma~\ref{lem:dirichletcurve}.
\end{proof}

For the thick--thin decomposition algorithm we will also need to be able to compute the distance from domain boundaries to the Dirichlet domain $\tD$.
With a bit of care, this works similarly to how it would in the Euclidean plane.

\begin{lemma}\label{lem:computedist}
    We can compute the minimum distance between a convex $n$-gon $C$ and a line segment $s$ in $\bigO{\log n}$ time in $\Hyp^2$, assuming we are given the vertices of $C$ in clockwise~order.
\end{lemma}
\begin{proof}
    Using binary search, we can in $\bigO{\log n}$ time find the chain $e_1, \dots, e_k$ of edges of $C$ visible from $s$.
    We only care about this chain, as a point in $C$ that is not visible from $s$ cannot yield the minimum distance.

    Note that the function that maps $i = 1, \dots, k$ to the minimum distance of $e_i$ to $s$ is convex.
    To see this, let $e_i^*$ for any $i$ denote the point on $e_i$ closest to $s$ and consider the line segment $e_{i-1}^* e_{i+1}^*$ for $i = 2, \dots, k-1$.
    Any line segment must achieve its maximum distance to $s$ at an endpoint, which means that in order for $e_i^*$ to be visible from $s$ it must lie closer to it than the farthest of $e_{i-1}^*$ and $e_{i+1}^*$.
    Hence there can be no local maximum at $i$, which implies convexity and lets us use a second binary search to find the distance of $C$ to $s$ in $\bigO{\log n}$ time.
\end{proof}

For an isometry $\gamma$ of $\Hyp^2$, its \emph{translation length} $\tau$ is the shortest distance that $\gamma$ can move a point, i.e.\ $\inf_{p \in \Hyp^2} |p \, \gamma(p)|$.
If $\tau > 0$ and $\gamma$ is orientation-preserving, then we call $\gamma$ a translation.
Any translation preserves exactly one line, which is its axis, as well as each hypercycle around that axis.
For a surface $\Hyp^2 / \Gamma$, note that we required $\Gamma$ to be a discrete group of isometries, which means that no element other than the identity may have translation length zero and thus the remaining ones are translations.
A translation $\gamma \in \Gamma$ corresponds to a closed geodesic of length $\tau$ and thus a $(\tau/2)$-thin part.
We prove the following for completeness.

\begin{lemma}\label{lem:translationlength}
    The length $\tau$ of any translation in $\Hyp^2$ can be calculated in time $\bigO{1}$. More precisely, in $\bigO 1$ time we can compute the strictly monotone function $\cosh(\tau/2)$, which allows us to compare translation lengths of different translations. 
\end{lemma}
\begin{proof}
    The group of orientation-preserving isometries of $\Hyp^2$ is isomorphic to the quotient $PSL(2,\Reals) = SL(2,\Reals) / \{I,-I\}$, where the special linear group $SL(2,\Reals)$ is the subgroup of matrices $\begin{pmatrix}a & b\\c & d\end{pmatrix}$ with determinant $ad-bc = 1$ \cite{katok1992fuchsian}.
    Important to note is that the isometry of $\Hyp^2$ corresponding to such a matrix is unrelated to $\begin{pmatrix}x \\ y\end{pmatrix} \mapsto \begin{pmatrix}a & b\\c & d\end{pmatrix} \begin{pmatrix}x \\ y\end{pmatrix}$.
    If we are working in the half-plane model with coordinates given by complex numbers $z$ with positive imaginary part (note that the choice of model is irrelevant since we can convert between models~\cite{cannon1997hyperbolic}), then the matrices $\pm\begin{pmatrix}a & b\\c & d\end{pmatrix} \in PSL(2,\Reals)$ correspond to (Möbius) transformations $z \mapsto \frac{az + b}{cz + d}$.


    The corresponding matrices for translations are diagonalisable \cite{katok1992fuchsian}.
    Combined with the restriction $ad-bc=1$, this lets us write them as $A = PDP^{-1}$ where $D =  \pm\begin{pmatrix}a & 0\\0 & 1/a\end{pmatrix}$ and $P \in PSL(2,\Reals)$.
    Matrix $D$ corresponds to $z \mapsto a^2 z$, which is the translation that preserves the line given by the imaginary axis and moves points along this line a distance $\tau = 2 \ln|a|$. Let $f,g: \Hyp^2 \to \Hyp^2$ denote the isometries represented by $D$ and $P$, respectively.
    The translation given by $A$ has the same translation length as $D$, since
    \[
        \tau
        = \inf_{p \in \Hyp^2} \left| p \, f(p) \right|
        = \inf_{p \in \Hyp^2} \left| g^{-1}(p) \, f(g^{-1}(p)) \right|
        = \inf_{p \in \Hyp^2} \left| p \, g(f(g^{-1}(p))) \right|.
    \]
    Both $A$ and $D$ also have the same trace, which is $\pm (a + 1/a) = \pm 2 \cosh(\tau / 2)$.
    Thus, we can calculate $\cosh(\tau / 2)$ from the trace of $A$.
\end{proof}

Finally, we need the following point-location data structure by Chan and Neckrich~\cite{ChanN18}, where we let $\eps=1$ and use the remark from their conclusion that it generalises to $x$-monotone curves.

\begin{theorem}[Theorem 3.8 of~\cite{ChanN18}]\label{thm:pointlocation}
    There is a point-location data structure for $n$ non-intersecting $x$-monotone curves in $\Reals^2$ that uses $\bigO{n}$ space, has query time $\bigO{\log n}$ and allows insertions in $\bigO{\log^2 n}$ time.
\end{theorem}

Before going into the thick-thin decomposition algorithm, let us consider the regions in $\Hyp^2$ that lift to a thin part on the surface.
For each thin part, there are infinitely many regions bounded by two hypercycles that lift to that thin part; we will call these \emph{bananas}.
Note that these bananas never intersect one another, even if they belong to different thin parts.

\thmthickthin*
\begin{proof}
    We first compute the thick-thin decomposition for $\mu_M$.
    The algorithm works by growing a value $\delta$ from $0$ to $\mu$ and considering the copies from $\Gamma \tD$ that have a point within distance $\delta$ from~$\tD$.
    We let $\widetilde \cD_\delta \subset \Hyp^2$ denote the region covered by all such copies and let $\partial \widetilde \cD_\delta \subset \Hyp^2$ be its boundary.
    Concretely, we maintain
    \begin{itemize}
        \item a point-location data structure $B$ containing the found bananas that intersect $\tD$ (note that there can be multiple for one thin part).
        More precisely, as we work in the Poincaré half-plane model of $\Hyp^2$, the two hypercycles bounding these bananas can be split into four Euclidean $x$-monotone curves, which lets us use the data structure of Theorem~\ref{thm:pointlocation}.
        \item a priority queue $Q$ containing the edges $e$ forming the boundary of $\partial \widetilde \cD_\delta$ that are not completely covered by a region from $B$ and have distance $\delta_e \leq \mu_M$ from $\tD$,
        with sorting based on $\delta_e$.
        \item a (self-balancing) binary search tree to remember the domains that have been processed.
    \end{itemize}    
    In every step of the algorithm, we extract an edge $e$ from $Q$ with minimum distance until $Q$ is empty.
    Let $\gamma \in \Gamma$ be such that $\gamma \tD$ is the domain copy incident to $e$ that has not been processed yet (if none exists we move on) and add $\gamma \tD$ to the processed domains.
    We calculate the translation distance $\tau$ of $\gamma$ using Lemma~\ref{lem:translationlength} and add the banana corresponding to $\gamma$ to $B$ when $\tau \leq 2\mu_M$.
    Finally, we add to $Q$ the $\bigO{g}$ edges of $\gamma \tD$ on the boundary of $\widetilde\cD_\delta \cup \gamma \tD$ that are not completely covered by regions from $B$ and have distance at most $\mu_M$ from $\tD$.

    By Lemma~\ref{lem:thintersect}, for any banana, we will consider at most $\bigO{g}$ domains that intersect it before seeing a domain $\gamma \tD$ that contains a point $\gamma \tilde p$ within small distance of the point $\tilde p \in \tD$.
    After that, adding the banana to $B$ avoids exploring it further.
    Since there are $\bigO{g}$ thin parts, each of which has $\bigO{g}$ associated bananas intersecting $\tD$ (also by Lemma~\ref{lem:thintersect}), we explore $\bigO{g^3}$ domains in thin parts.
    Note also that the bananas are disjoint and any domain where the corresponding translation has length at most $\mu_M$ must have a point within distance $\mu_M$ from~$\tD$.
    Finally, we also consider $\bigO{g^3}$ domains in thick parts by Lemma~\ref{lem:thicktersect}.

    For each of the $\bigO{g^3}$ domains, we consider $\bigO{g}$ edges and spend $\bigO{\log g}$ time for each.
    For each of the $\bigO{g^2}$ found bananas, we spend $\bigO{\log^2 g}$ time to add it to $B$.
    Thus, we get the total running time of $\bigO{g^4 \log g}$.
    Note that from the $\mu_M$-thick-thin decomposition we can also derive it for smaller $\mu$.
\end{proof}

\section{Steiner spanners for hyperbolic surfaces}
Using the thick-thin decomposition, we will now construct Steiner spanners.
The following constant-diameter constructions, proved in Section~\ref{sec:constdiam} are also key.

\thmconstdiam*

Additionally, we need the following approximation result, proved in \cite{Kisfaludi-BakW25} which will let us place Steiner points at distance $\sqrt\eps$ rather than $\eps$.

\begin{lemma}[Lemma 4 of \cite{Kisfaludi-BakW25}]\label{lem:split}
    Let $p,q,x \in \Hyp^d$ and let $y$ be the closest point to $x$ on the segment $pq$.
    Assume $\eps \in (0,1]$ and $|pq| \leq \Delta$ for some $\Delta > 0$.
    If $|xy| \leq \sqrt \eps \min\{|py|,|qy|\}$, then $|px| + |xq| \leq (1 + e^\Delta \eps) |pq|$.
\end{lemma}

The following lemma shows that it suffices to cut open the thin parts in two different ways and handle both separately.

\begin{lemma}\label{lem:thinsplit}
    Let $S$ be a $\mu$-thin part of a hyperbolic surface for $\mu \leq \mu_M$.
    Then, there are two simply-connected regions such that any pair of points on $S$ has a shortest curve contained entirely within one of the regions.
\end{lemma}
\begin{proof}
    Let $l_1$ and $l_2$ be shortest curves connecting the two components of $\partial S$ such that $l_1$ and $l_2$ have maximum distance from each other.
    We get one of the regions from the lemma statement by cutting open $S$ along $l_1$ and another by cutting along $l_2$ (note that each covers all of $S$ by itself).
    
    For the proof, we will first cut $S$ along both $l_1$ and $l_2$ to partition it into isometric parts $S_1$ and $S_2$. We will show that any curve can be made to only intersect one of $l_1$ or $l_2$ without increasing its length.
    Without loss of generality, assume that $s$ crosses from $S_1$ to $S_2$ at $p_1 \in l_1$ and crosses back to $S_1$ at $p_2 \in l_2$.
    Since $S_1$ and $S_2$ are isometric, we can replace the segment of $s$ between $p_1$ and $p_2$ with one that lies on $S_1$ and has the same length.
    This removes the crossings $p_1$ and $p_2$ from $s$ and we can repeat this procedure until no such pair $p_1,p_2$ remains, proving the lemma.
\end{proof}

The final ingredient for our Steiner spanners is given by $\eps$-nets.
In particular, we first give a polynomial-time algorithm for computing an $\eps$-net for only the thick parts of a surface (which Lanuel~\cite{lanuelthesis} refers to as a pseudo-$\eps$-net).
It is worth noting that we improve Lanuel's result, whose running time is $\bigO{1/\eps^4}$ with a hidden exponential dependency on $g$.

\begin{lemma}\label{lem:surfnet}
    Given a hyperbolic surface $S$ as a Dirichlet domain $\tD$ and values $0 < r \leq \mu \leq \mu_M$, we can construct a point set of size $\bigO{g / r^2}$ that is an $(r/3)$-packing and an $r$-cover of the $\mu$-thick parts of $S$ in $\bigO{(g^4/r^2)\log(g/r)}$ time.
\end{lemma}
\begin{proof}
    First, use Theorem~\ref{thm:thickthin} to get a thick-thin decomposition for $\tD$ and let $\cD$ be the neighbouring set of copies of $\tD$ given by Lemma~\ref{lem:thicktersect}.
    To ensure we do not place net points too closely together, we use a dynamic (approximate) nearest neighbour data structure for $\Hyp^2$.
    For example, for $(3/2)$-approximate nearest neighbours we can get a data structure that uses $\bigO{n}$ space and allows queries and updates in $\bigO{\log n}$ time with $n$ points~\cite{steinerspanner}.
    Call this data structure $A$; we will use it as follows.
    Whenever we add a point to our net $\widetilde N \subseteq \tD$, we also add its $\bigO{g^3}$ copies in the other domains of $\cD$ to $A$.
    Then, to check if we actually want to add a point $\tilde p$ to $\widetilde N$, we query $A$ for a $(3/2)$-approximate nearest neighbour $\tilde q$ and see if $|\tilde p \tilde q| > r/2$.

    Note that this ensures that $\widetilde N$ is an $(r/3)$-packing, because a $(3/2)$-approximate nearest neighbour having distance greater than $r/2$ implies the exact nearest neighbour has distance greater than $r/3$.
    Thus, we can make a collection of balls of radius $r/3$ around the net points, where the balls will be disjoint and topological disks.
    This makes each contribute area $4\pi^2 \sinh^2(r/6) = \bigOm{r^2}$, while the surface has total area $\bigTh{g}$, meaning $|\widetilde N| = \bigO{g / r^2}$.
    Consequently, inserting all points into $A$ takes $\bigO{(g^4 / r^2) \log(g/r)}$ time, $A$ uses $\bigO{g^4 / r^2}$ space, and queries take $\bigO{\log(g/r)}$ time.

    Let $Q_r$ be as in Theorem~\ref{thm:quadtiling}.
    For each tile $\widetilde T \in Q_{r/2}$ intersecting a region $\widetilde R \subseteq \tD$ marked as thick, we take an arbitrary point $\tilde p \in \widetilde R \cap \widetilde T$ and try to add it to $\widetilde N$.
    This gives $\bigO{g^4 / r^2}$ candidates, as these tiles contain balls of radius $\bigOm{r}$ and thus have area $\bigOm{\sinh^2 r} = \bigOm{r^2}$, while also being contained in the region covered by $\cD$ which has area $\bigO{g^4}$.
    As we spend $\bigO{\log(g/r)}$ time per unsuccessful candidate, the total running time remains $\bigO{(g^4 / r^2) \log(g/r)}$.

    Lastly, by construction, the candidate points give an $r/2$-cover and any point that was within distance $r/2$ from an unsuccessful candidate point has to be within distance $r$ from a net point, making $\widetilde N$ an $r$-cover.
\end{proof}

\paragraph{Construction.}
First, use Theorem~\ref{thm:thickthin} to get a thick-thin decomposition based on the Margulis constant~$\mu_M$.
We start by constructing a set $C$ of centres and $P_2$ of Steiner points around those centres.
Throughout the construction, there will be four relevant radii around centres $c \in C$:
radius $\frac{1}{48}$ is used for packings,
radius $\frac{1}{16} = \frac{2}{32}$ for coverings,
radius $\frac{3}{32}$ for placing Steiner points in $P_2$,
and radius $\frac{1}{8} = \frac{4}{32}$ for the constant-diameter Steiner spanner construction of Theorem~\ref{thm:constdiam} (but, note that we do not need these exact values, so we can approximate them as required).
In particular, restricted to the thick parts, $C$ will be the $\frac{1}{16}$-cover given by Lemma~\ref{lem:surfnet}.
For each of these centres $c \in C$, we then add Steiner points to $P_2$ spaced at distance $\sqrt\eps$ along a circle of radius $\frac{3}{32}$ around~$c$.

For each thin part, since the shape is completely determined by the length of its shortest loop, we can construct and work on the two simply-connected regions in $\Hyp^2$ given by Lemma~\ref{lem:thinsplit}, then later map the points onto $\tD$.
Let $Q_{\frac{1}{16}}$ be as in Theorem~\ref{thm:quadtiling}.
In each such region $R \subset \Hyp^2$, we place a centre for each tile of $Q_{\frac{1}{16}}$ that intersects $R$ as long as the ball of radius $\frac{3}{32}$ around it has perimeter at least $\sqrt\eps$ in the interior of $R$.
For each centre $c$ placed for $R$, we take the intersection of the circle of radius $\frac{3}{32}$ around $c$ with $R$ and add Steiner points to $P_2$ spaced at distance $\sqrt\eps$ along this intersection.
Additionally, we greedily add centres to $C$ to ensure every point in $P$ has a centre within distance $\frac{1}{16}$ (the details for this are given in the theorem proof), and add a Steiner point on both boundary components of the ball of radius $\frac{3}{32}$ around each centre $c \in C$ added in this step.

Now, for each centre $c \in C$ in the thick part, use Theorem~\ref{thm:constdiam} to construct a Steiner $(1+\eps)$-spanner for all points of $P \cup P_2$ in the ball $B$ of radius $\frac{1}{8}$ around $c$ (details on how to find these points are given in the theorem proof) and replace every newly placed Steiner point with its closest point inside $B$.
Note that this gives point sets of diameter at most $\frac{1}{4}$ so we may in fact use Theorem~\ref{thm:constdiam}.
Moreover, $\frac{1}{8} < \mu_M / 2$ so balls of this radius in a thick part must be geodesically convex (i.e.\ it cannot be the case that the shortest curve between two points in the ball leaves the ball and goes around the surface in a different way).
We do the same for centres $c \in C$ in the thin part, except here we take the ball $B$ in the simply-connected region where $c$ was originally placed so that Theorem~\ref{thm:constdiam} works on a simply-connected region.
Finally, consider the regions left uncovered by the balls of radius $\frac{3}{32}$ around centres in $C$.
Each of these regions has two boundary components that both contain a Steiner point from~$P_2$; we connect these with a single edge each.
This adds $\bigO{n + g}$ edges.
We now have our graph $G$, which we will prove to be the promised Steiner spanner.

\thmsurfspanner*
\begin{proof}
    We will assume without loss of generality that $\eps>0$ is below a sufficiently small constant threshold.
    Let us first bound the sizes of $C$ and $P_2$.
    The thick parts contribute $\bigO{g}$ centres in total, each of which adds $\bigO{1/\sqrt\eps}$ Steiner points to $P_2$.
    For thin parts, note that by Lemma~\ref{lem:thinjradius} the injectivity radius decreases exponentially.
    Thus, we add $\bigO{\log\frac1\eps}$ centres for each thin part but the number of Steiner points placed follows a geometric series, giving only $\bigO{1 / \sqrt\eps}$ Steiner points per thin part.
    Finally, we place a constant number of centres and Steiner points for each point from $P$ that was not covered yet.
    This gives final sizes $|C| = \bigO{n + g \log\frac{1}{\eps}}$ and $|P_2| = \bigO{n + g / \sqrt{\eps}}$.

    Each of the balls $B$ of radius $\frac{1}{8}$ around a centre only intersects a constant number of others, based on an area argument.
    Namely, for any point $c \in C$, the balls that intersect its ball must come from a point of $C$ within distance $\frac{1}{4}$.
    A ball of radius $\frac{1}{4} + \frac{1}{32}$ has area $a_1 = \bigTh{1}$ in $\Hyp^2$ and at most area $a_1$ in $P_2$.
    Balls of radius $\frac{1}{48}$ are still topological disks in the thick parts of $P_2$ and thus have the same area $a_2 = \bigTh{1}$ as in $\Hyp^2$.
    Now, since $C$ is a $\frac{1}{48}$-packing, this means we can have at most $a_1/a_2 = \bigO{1}$ points of $C$ within distance $\frac{1}{4}$ from $p$.
    The balls have constant ply by similar reasoning.

    Since $|P \cup P_2| = \bigO{\frac{g}{\sqrt\eps} + n}$ and each of these points is involved in $\bigO{1}$ small-diameter Steiner spanner constructions from Theorem~\ref{thm:constdiam}, we get either $\bigO{(\frac{g}{\sqrt\eps} + n) \frac{1}{\sqrt{\eps}}}$ or $\bigO{(\frac{g}{\sqrt\eps} + n) / \eps^{3/2}}$ edges, depending on whether intersections are allowed.
    Note that the complete graph has size $n^2$ (or $n^4$ when we add Steiner points to make it non-crossing), so we can assume that we only use this procedure when it gives an edge count bounded by $n^{\bigO{1}}$.
    In particular, this lets us assume $\frac{1}{\eps} = n^{\bigO{1}}$ and $g = n^{\bigO{1}}$.

    \medskip
    Before we bound the running time, let us fill in the missing algorithmic details.
    To get, for each centre from $C$, the points of $P \cup P_2$ within distance $\frac{1}{8}$, we build a point-location data structure on the balls around these centres.
    For the thick parts, we first take the $\bigO{g^3}$ copies of each centre given by Lemma~\ref{lem:thicktersect} and around each of these $\bigO{g^4}$ points we take the ball of radius $\frac{1}{8}$.
    For the thin parts, we take the ball of radius $\frac{1}{8}$ inside the simply-connected region the centre was placed for, then take the $\bigO{g}$ copies of this given by Lemma~\ref{lem:thintersect} for $\bigO{g^2 \log\frac1\eps}$ copies in total.
    We then take the arrangement of these regions around centres.
    Note that, because each ball only intersects a constant number of others, the arrangement has size $\bigO{g^4 + g^2 \log\frac1\eps}$ and as we are working in the half-plane model consists of Euclidean circle arcs and line segments.
    For this arrangement, we use a static point-location data structure that can handle circle arcs~\cite{EdelsbrunnerGS86,SarnakT86}, which takes $\bigO{(g^4 + g^2 \log\frac1\eps) \log(g\log\frac1\eps)}$ time to construct, uses $\bigO{g^4 + g^2 \log\frac1\eps}$ space, and allows queries in $\bigO{\log(g\log\frac1\eps)} = \bigO{\log n}$ time.
    For every point in $P \cup P_2$ we can now query in which balls it lies and then add it to a list for these balls, taking $\bigO{n \log n}$ time in total.
    
    For the step where we ensure that every point in $P$ has a centre from $C$ within distance $\frac{1}{16}$, we first use the same setup as above except with balls of radius $\frac{1}{16}$ to find the points $P_{\text{uncov}} \subseteq P$ still to be covered.
    In particular, we get a partition of $P_{\text{uncov}}$ based on the thin parts.
    For each thin part, we sort the points based on their distance along the thin part, then the algorithm repeatedly adds a centre $c$ to $C$ at the leftmost point and deletes the subsequent points that are also within distance $\frac{1}{16}$ from that centre.
    Although sorting along the thin part is not exactly the same as sorting by distance from $c$, we can safely say that any remaining point has distance at least $\frac{1}{48}$ to $c$, which gives us the same packing property as before.
    This takes $\bigO{n \log n}$ time.
    
    For the small-diameter Steiner spanner constructions from Theorem~\ref{thm:constdiam}, note that
    the total number of points involved in the construction is $\bigO{|P \cup P_2|} = \bigO{\frac{g}{\sqrt\eps} + n}$ (as a point can appear in $\bigO{1}$ constructions).
    Thus, this takes $\bigO{(\frac{g}{\sqrt\eps} + n) \log(\frac{g}{\sqrt\eps} + n) / \sqrt\eps} = \bigO{(\frac{g}{\eps} + \frac{n}{\sqrt\eps}) \log n}$ time when we allow edge intersections and $\bigO{(n / \eps^{3/2} + g/\eps^2) \log n}$ time when we do not.
    
    \medskip
    Finally, let us consider two points $p, q \in P$ to prove that this is in fact a Steiner $(1+\bigO{\eps})$-spanner.
    First, if there is a centre $c \in C$ around which the ball $B$ of radius $\frac{1}{8}$ contains both $p$ and $q$, then we are done as we have already constructed a Steiner $(1+\eps)$-spanner for the points in $B$ according to Theorem~\ref{thm:constdiam}.
    Otherwise, we will find a sequence of Steiner points $s_1, \dots, s_k \in P_2$ with also $s_0 = p$ and $s_{k+1} = q$, where for every $i = 1,\dots,k+1$ the point $s_i$ has distance at most $\sqrt\eps$ from $pq$ and distance at least $\frac{1}{32}$ from $s_{i-1}$ while also having a path of length at most $(1+\eps)|s_{i-1}s_i|$ in $G$ connecting it to $s_{i-1}$.
    Note that with Lemma~\ref{lem:split} this implies that $G$ is a Steiner $(1+\bigO{\eps})$-spanner, so we can get a Steiner $(1+\eps)$-spanner by following the construction with an appropriately adjusted value $\eps$.

    First, consider each occasion that $pq$ stops being covered by the balls of radius $\frac{1}{8}$ around the centres in $C$.
    For one such case, let $c \in C$ be the centre of the last ball intersected before $pq$ stops being covered and let $c' \in C$ be the centre of the first ball after this.
    There will be Steiner points $s_i, s_{i+1} \in P_2$ connected by an edge in $G$, placed at distance $\frac{3}{32}$ from $c$ and $c'$ respectively.
    Since this can only happen in a $\sqrt\eps$-thin part, $s_i$ and $s_{i+1}$ must be within distance $\sqrt\eps$ from $pq$.
    Finally, $|s_i s_{i+1}| \geq \frac{1}{16}$ as both have distance at least $\frac{1}{32}$ from where $pq$ intersects the ball of radius $\frac{1}{8}$ around their respective centre.

    Next, we continue by finding a point $s_{i+1}$ given $s_i$.
    Note that at this point we have already fixed a subsequence of points that we must go through, but we first assume that there is no centre $c \in C$ that has within distance $\frac{1}{8}$ both $s_i$ and a point that was already selected to come after $s_i$.
    Consider the ball $B$ of radius $\frac{1}{16}$ whose centre lies on $pq$ that has $s_i$ on its boundary.
    There has to be a centre $c \in C \cap B$ for $C$ to be a $\frac{1}{16}$-cover, and $|s_i c| \leq \frac{1}{8}$.
    There also has to be a Steiner point $s_{i+1} \in P_2$ at distance $\sqrt\eps$ from $pq$ and distance $\frac{3}{32}$ from $c$.
    Note that the circle of radius $\frac{1}{16}$ around $c$ would have an intersection with $pq$ at distance at least as far from $p$ as the centre of $B$, so this must also hold for the circle of radius $\frac{3}{32}$ around~$c$.
    Hence $|s_i s_{i+1}| \geq \frac{1}{16} - \sqrt\eps$.

    Last, we consider the case where there is a centre $c \in C$ whose ball $B$ of radius $\frac{1}{8}$ contains both $s_i$ and a point that was already selected to come after $s_i$, which we now mark $s_{i+1}$.
    Since $s_i, s_{i+1} \in B$ the Steiner $(1+\eps)$-spanner of Theorem~\ref{thm:constdiam} constructed for $B$ connects them with a path of length at most $(1+\eps) |s_i s_{i+1}|$.
    What remains is to prove $|s_i s_{i+1}| \geq \frac{1}{32}$.
    If $s_i$ was found by the procedure in the preceding paragraph, let $c \in C$ be the centre used there.
    We know that $|c s_i| = \frac{3}{32}$ while $|c s_{i+1}| \geq \frac{1}{8}$, so this gives the desired bound $|s_i s_{i+1}| \geq \frac{1}{32}$.
    Otherwise, if $s_i = p$, then $s_{i+1}$ has to be placed at distance $\frac{3}{32}$ from the closest centre in $C$, while $p$ has to have distance at most $\frac{1}{16}$ to its closest centre, giving $|p s_{i+1}| \geq \frac{1}{32}$.
    The same argument holds when $s_{i+1} = q$.
    In the remaining case, note that the minimum distance between the two boundary components of a ball in the $\sqrt\eps$-thin part is already greater than $\frac{1}{16}$ which gives the bound.

    This finishes the proof that $G$ is a Steiner $(1+\eps)$-spanner and with that the proof of the theorem.
\end{proof}

\section{Non-crossing Steiner spanner for the hyperbolic plane}
Using the same ideas as for hyperbolic surfaces, we will now also construct a non-crossing Steiner spanner for points in the hyperbolic plane.
Given a polygon $P$, a \emph{neck} is a triangle or quadrilateral $N\subseteq P$ with two sides contained in $\partial P$; every remaining side has length at most $1$, and, in the quadrilateral case, the two remaining sides are non-adjacent.
Now, a \emph{neck decomposition} is a set of disjoint necks where no new neck can be inserted without losing disjointness and no pair of necks can be replaced by a single larger neck that contains both.

\begin{lemma}\label{lem:neckcomp}
    Any neck decomposition of a convex $n$-gon $P$ has size at most $2n-3$ and we can find one in $\bigO{n}$ time.
\end{lemma}
\begin{proof}
    First, note that a neck decomposition can contain at most one neck per pair of edges.
    We will now use induction to bound the number of necks.
    For $n = 3$, there are three pairs of edges so it holds trivially.
    Next, we consider what happens when we remove a vertex $v$ from $P$ and connect its neighbours $u, w$ to get a smaller polygon $P'$.
    By the induction hypothesis, any neck decomposition of $P'$ has at most $2n-5$ necks.
    In $P$, all necks that did not involve $uw$ remain the same and we can also add a (triangular) neck between $uv$ and $vw$.
    Any other neck in $P$ involving $uv$ or $vw$ must have come from a neck involving $uw$ in $P'$.
    In particular, at most one neck of $P'$ involving $uw$ can split into two necks of $P$, one involving $uv$ and one involving $vw$.
    Thus, a neck decomposition of $P$ can have at most two more necks than a neck decomposition of~$P'$, giving at most $2n-3$ necks as promised.

    We can now use a recursive algorithm to find all pairs of edges that have a neck between them.
    For $n \leq 3$ we can simply check all pairs of edges.
    Otherwise, we will implicitly use the Dobkin-Kirkpatrick hierarchy~\cite{DobkinKirkpatrick} of $P$:
    we start by removing a set $\Delta V$ of $\lfloor n/2 \rfloor$ non-adjacent vertices of $P$ and connecting the vertices that were adjacent to the removed ones to get a smaller polygon $P'$.
    We then compute a neck decomposition $\cN'$ for $P'$, which has size at most $2\lceil n/2 \rceil - 3$.
    Now, each pair of edges in $P'$ corresponds to at most four edges in~$P$.
    For these, we check which of the six pairs still have a neck between them.
    Additionally, we add a (triangular) neck at each vertex in $\Delta V$.
    Thus, we check at most $\lfloor n/2 \rfloor + 12\lceil n/2 \rceil - 18$ pairs in this step,  meaning in the whole recursion we check less than $\sum_{i=0}^\infty 7n/2^i = \bigO{n}$ pairs.
\end{proof}

Removing these necks leaves us with more well-behaved polygons.

\begin{theorem}
    Given a convex $n$-gon $P$ and a constant $\eps > 0$, we can find a neck decomposition $\cN$ of $P$ and an $\eps$-cover for $P' = P \setminus \bigcup_{N \in \cN} N$ of size $\bigO{n}$ in $\bigO{n}$ time.
\end{theorem}
\begin{proof}
    First calculate $\cN$ according to Lemma~\ref{lem:neckcomp}.
    Next, let $\eps' = \min\{\eps, \frac{1}{2}\}$ and let $Q_{\eps'}$ be the tiling given by Theorem~\ref{thm:quadtiling}.
    For each tile $T \in Q_{\eps'}$ that intersects the non-neck part $P'$, we add an arbitrary point from $T \cap P'$ to our $\eps$-cover $S$.
    Note that $S$ is an $\eps$-cover by construction, so what remains to show is that $|S| = \bigO{n}$.

    First, note that there are $\bigO{n}$ tiles that intersect a neck or are adjacent to a tile that does.
    Now consider a tile $T$ from the remainder and assume (for contradiction) that $T$ and a neighbour $T'$ are both not fully contained in $P$ and intersect different boundary segments.
    We can make a quadrilateral $N$ from points $p,q \in T \cap \partial P$ and $p',q' \in T' \cap \partial P$, which will be a neck since tiles of $Q_{\eps'}$ have diameter less than $\frac{1}{2}$.
    The only reason why $N \notin \cN$ would be if it intersected a neck $N' \in \cN$, but either way there would be a neck intersecting $T \cup T'$ to contradict our initial assumption.
    Thus, the tile $T$ must have an adjacent tile $T'$ completely contained in $P$.
    As $P$ has area $\bigO{n}$, it can contain at most $\bigO{n}$ tiles of $Q_{\eps'}$, each of which has $\bigO{1}$ adjacent tiles.
    Hence, $|S| = \bigO{n}$.
\end{proof}

To now construct a non-crossing Steiner spanner, we first compute the convex hull $C$ of our point set, as it suffices to construct the spanner inside $C$.
Note that we can use any Euclidean convex hull algorithm for this, after we convert the point set to the Beltrami-Klein disk model.
In this model, hyperbolic geodesics appear as Euclidean line segments, and thus the Euclidean convex hull is also the hyperbolic convex hull.
As a consequence, we can find $C$ in $\bigO{n \log n}$ time using well-known algorithms such as Graham's scan~\cite{grahamscan}.\footnote{In fact, all required operations for Graham's scan can also be done directly in the half-plane model to avoid the conversion. This is necessary as the conversion would result in precision issues.}
Next, we compute the neck decomposition for $C$.
After this, the details follow those for Steiner spanners on hyperbolic surfaces, with necks replacing thin parts.

In the analysis we can now note that the $\bigO{g}$ bounds that came from the area of the surface or the number of thin parts are replaced by $\bigO{n}$ bounds for the area and number of necks, which gives $\bigO{n^2 / \eps^2}$ edges in total.
For the running time, we no longer need to concern ourselves with different copies of points, so all auxiliary operations together take $\bigO{n \log n \log\frac1\eps}$ time and constructing the small-diameter Steiner spanners as in Theorem~\ref{thm:constdiam} takes $\bigO{n \log n / \eps^2}$ time, which dominates.

\thmplanespanner*

\section{Steiner spanners for constant-diameter regions}\label{sec:constdiam}
The hyperbolic plane is locally Euclidean, meaning that any subset of the hyperbolic plane of sufficiently small diameter can be embedded into the Euclidean plane with small distortion.
In practice, this means that many Euclidean constructions still work for sets of constant diameter, though using these constructions as a black box only gives a $\bigO{1}$-spanner rather than a $(1+\bigO{\eps})$-spanner.
We will generalise the construction of \cite{planesteinernew}, which comes with the additional challenge that it heavily relies on rotations, which do not behave as nicely in the hyperbolic plane.
To keep the same edge intersections, we will use the Klein disk, which keeps lines straight but can heavily distort both distances and angles.
Note that we can convert from any model of $\Hyp^2$ to the Klein disk~\cite{cannon1997hyperbolic}, but these conversions involve a square root which our real-RAM model does not have access to.
However, this poses no issue, as with Steiner spanners we can afford small errors and we can check for edge intersections in the original model.
To start, we bound these distortions by a function of the distance from the origin.

\begin{lemma}\label{lem:curvelength}
    Let $\kappa$ be a curve in the Klein disk of Euclidean length $L_\Euc$ and hyperbolic length $L_\Hyp$.
    If $\kappa$ has maximum hyperbolic distance $R_\Hyp$ from the origin, then $L_\Euc \leq L_\Hyp \leq e^{2R_\Hyp} L_\Euc$.
\end{lemma}
\begin{proof}
    Let $R = \tanh(R_\Hyp)$ be the Euclidean maximum distance of $\kappa$ to the origin.
    We will prove the statement using the metric tensor, which generalises the Euclidean inner product:
    for any point $x$ in the Klein disk, the metric tensor is a bilinear form $g_x(.,.)$ that takes two tangent vectors $v_x^1, v^2_x$ at $x$ to a value in $\Reals$.
    A tangent vector $v_x$ has hyperbolic length $\sqrt{g_x(v_x, v_x)}$. We note that angles at $x=0$ are conformal, so we consider $x\neq 0$,
    We can write $x$ and $v_x$ as vectors in $\Reals^2$, so that the metric tensor gives \cite{cannon1997hyperbolic}
    \[ g_x(v_x,v_x) = \frac{\|v_x\|^2}{1 - \|x\|^2} + \frac{(x \cdot v_x)^2}{(1 - \|x\|^2)^2} = \|v_x\|^2 \cdot \frac{1 - \|x\|^2 + \|x\|^2 \left( \frac{x}{\|x\|} \cdot \frac{v_x}{\|v_x\|} \right)^2}{(1 - \|x\|^2)^2}. \]
    Seeing as $0 \leq \left( \frac{x}{\|x\|} \cdot \frac{v_x}{\|v_x\|} \right)^2 \leq 1$, we get $\frac{\|v_x\|^2}{1 - \|x\|^2} \leq g_x(v_x,v_x) \leq \frac{\|v_x\|^2}{(1 - \|x\|^2)^2}$, which becomes $\|v_x\| \leq \sqrt{g_x(v_x,v_x)} \leq \frac{\|v_x\|}{1 - R^2}$ when $\|x\| \leq R$.
    Assuming $\kappa$ is defined as a differentiable map $[0,1] \to \Reals^2$, we can now calculate curve length as
    \begin{align*}
         L_\Hyp = \int_0^1 \sqrt{g_{\kappa(t)}(\kappa'(t), \kappa'(t))}\ dt &\geq \int_0^1 \|\kappa'(t)\|\ dt = L_\Euc, \\
         L_\Hyp = \int_0^1 \sqrt{g_{\kappa(t)}(\kappa'(t), \kappa'(t))}\ dt &\leq \int_0^1 \frac{\|\kappa'(t)\|}{1 - R^2}\ dt = \frac{L_\Euc}{1 - R^2}.
    \end{align*}

    This proves $L_\Euc \leq L_\Hyp \leq \frac{L_\Euc}{1 - R^2}$.
    To get the lemma statement, note that $\frac{1}{1 - R^2} = \frac{1}{1 - \tanh^2(R_\Hyp)} = \cosh^2(R_\Hyp) \leq e^{2R_\Hyp}$.
\end{proof}

We can now similarly bound the distortion of angles using the distances and trigonometry.

\begin{lemma}\label{lem:angles}
    Let two lines in the Klein disk intersect at a Euclidean angle $\alpha_\Euc$ at a point $x$.
    If $x$ lies at hyperbolic distance $R_\Hyp$ from the origin and $\alpha_\Euc \leq \frac{1}{2}e^{-2R_\Hyp}$, then the hyperbolic angle between the lines is $\alpha_\Hyp \leq \pi e^{2R_\Hyp} \alpha_\Euc$.
\end{lemma}
\begin{proof}
    Add a third line to get a Euclidean right-angled triangle that has angle $\alpha_\Euc$ at $x$ and whose vertices are all at least as close to the origin as $x$.
    Let $\gamma_\Hyp$ denote the hyperbolic angle at the Euclidean right angle.
    Name the sides $a$, $b$ and $c$, where $a$ is opposite angle $\alpha$ and $c$ opposite the Euclidean right angle, and use subscripts $\Euc$ and $\Hyp$ for the Euclidean and hyperbolic side lengths.
    We can move the third line close enough so that $a_\Hyp \leq 2.1$.
    By the hyperbolic law of sines \cite[page 81]{thurston97three},
    \[ \frac{\sin\alpha_\Hyp}{\sinh a_\Hyp} = \frac{\sin\gamma_\Hyp}{\sinh c_\Hyp}. \]
    Note that $\sin\gamma_\Hyp \leq 1$, that $\sinh c_\Hyp \geq c_\Hyp$, and that $a_\Hyp \leq 2.1$ implies $\sinh a_\Hyp \leq 2 a_\Hyp$.
    Thus,
    \[ \sin\alpha_\Hyp \leq 2 a_\Hyp / c_\Hyp. \]
    Applying Lemma~\ref{lem:curvelength} now gives
    \begin{align*}
        2a_\Hyp / c_\Hyp
        \leq 2e^{2R_\Hyp} \cdot a_\Euc / c_\Euc 
        = 2e^{2R_\Hyp} \sin\alpha_\Euc 
        \leq 2e^{2R_\Hyp} \alpha_\Euc.
    \end{align*}
    Note that $2e^{2R_\Hyp} \alpha_\Euc \leq 1$ by assumption.
    Thus,
    \[ 2e^{2R_\Hyp}\alpha_\Euc \leq \sin(\pi e^{2R_\Hyp}\alpha_\Euc), \]
    which gives us $\alpha_\Hyp \leq \pi e^{2R_\Hyp}\alpha_\Euc$ as claimed.
\end{proof}

Last, we will need this small auxiliary lemma.

\begin{lemma}\label{lem:tanhapprox}
    If $\eps \in [0,\frac{1}{4}]$ and $x \in [0,1]$, then
    \[ (1 + \eps) \tanh x \leq \tanh\left( (1+4\eps) x \right). \]
\end{lemma}
\begin{proof}
    Note that this holds trivially when $\eps = 0$ or $x = 0$.
    Otherwise, the following are all equivalent (using the product identities for hyperbolic functions):
    \begin{align*}
        (1 + \eps) \tanh x &\leq \tanh\left( (1+4\eps) x \right), \\
        (1 + \eps) \sinh x \cosh\left( (1+4\eps) x \right) &\leq \sinh\left( (1+4\eps) x \right) \cosh x, \\
        (1 + \eps) \left( \sinh\left( (2+4\eps) x \right) - \sinh(4\eps x) \right) &\leq \sinh\left( (2+4\eps) x \right) + \sinh(4\eps x), \\
        \eps \sinh\left( (2+4\eps) x \right) &\leq (2+\eps) \sinh(4\eps x).
    \end{align*}
    Here, using the sum of arguments identity,
    \[ \sinh\left( (2+4\eps) x \right) = \sinh(4\eps x) \cosh(2x) + \cosh(4\eps x) \sinh(2x). \]
    Thus, after dividing by $\sinh(4\eps x)$, the inequalities are also equivalent to
    \[ \eps\cosh(2x) + \eps\coth(4\eps x) \sinh(2x) \leq 2 + \eps. \]
    Now, since $x \leq 1$ while $\coth t \leq \coth(1) / t$ and $\sinh(2t) \leq \sinh(2) \cdot t$ for $t \in (0, 1]$, this inequality is implied by
    \begin{align*}
        \eps\cosh 2 + (\coth 1 \sinh 2) / 4 &\leq 2 + \eps \\
        \eps &\leq \frac{2 - (\coth 1 \sinh 2) / 4}{\cosh 2 - 1} \approx 0.29,
    \end{align*}
    which is true by our choice of $\eps$.
\end{proof}

The construction of \cite{planesteinernew} is a \emph{cone-restricted Steiner $(1+\eps)$-spanner}, meaning that, for every $a,b\in P$, there is an $ab$-path $\Pi$ in the spanner such that every edge $e$ of $\Pi$ has Euclidean angle $\angle(e, ab) \leq \sqrt{\eps}$.
To prove that such cone-restricted Steiner spanners also work for hyperbolic point sets, we now consider the \emph{hyperbolic projection} of edges of these spanners onto some base line.
The hyperbolic projection maps a point $p$ to the hyperbolically closest point $p'$ on a given line $\ell$, where $pp'$ is now also hyperbolically perpendicular to $\ell$.

\begin{lemma}\label{lem:conerestrict}
    Let $\ell$ and $\ell'$ be two lines that intersect (possibly outside of the Klein disk) at Euclidean angle $\alpha \leq \sqrt{\eps}$ for some $\eps \in (0,10^{-5}]$.\footnote{Note that we make no attempt to optimise this constant.}
    Let $p$ and $q$ be points on $\ell$ that hyperbolically project to $p'$ and $q'$ on $\ell'$.
    Assume each of $p$ and $q$ is at (hyperbolic) distance at most~$\frac{1}{2}$ from the origin and distance at most $\sqrt\eps$ from $\ell'$.
    Then, $|pq| \leq (1+\eps \cdot 10^5) |p'q'|$.
\end{lemma}
\begin{proof}
    \begin{figure}
        \centering
        \includegraphics[page=1]{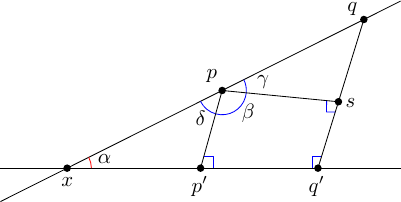}
        \includegraphics[page=2]{Figures/projectionfig.pdf}
        \caption{Construction in Lemma~\ref{lem:conerestrict}. Here, red denotes Euclidean angles while blue denotes hyperbolic (right) angles.}
        \label{fig:conerestrict}
    \end{figure}

%
    Assume without loss of generality that $|pp'| \leq |qq'|$.
    Let $\delta$ denote the (hyperbolic) angle between $pp'$ and $\ell$ on the side away from $q$; we will first need to prove that $\cos\delta \leq 100\sqrt{\eps}$.
    We handle this with two cases.
    \begin{itemize}
    \item
    This case is shown in Figure~\ref{fig:conerestrict} (left).
    Assume that there is an intersection $x$ between $\ell$ and $\ell'$, and that $|px| \leq 1$.
    In this case, we know that $p$ lies between $x$ and $q$.
    Let $\alpha_\Hyp$ be the hyperbolic angle at $x$ corresponding to $\alpha$.
    Hyperbolic trigonometry gives~\cite[page 81]{thurston97three}
    \[ \cos\delta = \cosh|p'x| \sin\alpha_\Hyp. \]
    Note that $\alpha_\Hyp \leq \pi e^3 \alpha \leq \pi e^3 \sqrt{\eps}$ by Lemma~\ref{lem:angles} and that $|p'x| \leq |px| \leq 1$, so we get
    \[ \cos\delta \leq \cosh 1 \cdot \sin(\pi e^3 \sqrt{\eps}) \leq \pi e^3\cosh 1 \cdot \sqrt{\eps} \leq 100\sqrt{\eps}. \]

    
    \item
    Otherwise, let $r_1$ be the point on $\ell$ at distance $1$ from $p$ such that $p$ is between $r_1$ and $q$, then let $r_2$ be the (hyperbolic) projection of $r_1$ onto $pp'$ (as shown in Figure~\ref{fig:conerestrict} (right)).
    If $\delta \geq \pi/2$ we are already done, so now the fact that $r_1$ is on the same side of $\ell'$ as $p$ implies that $r_2$ must lie between $p$ and $p'$.
    Using hyperbolic trigonometry,
    \[ \cos\delta = \tanh|pr_2| / \tanh|pr_1|. \]
    Note here that $|pr_2| \leq |pp'| \leq \sqrt\eps$ and $|pr_1| = 1$, so we get
    \[ \cos\delta \leq \tanh\sqrt{\eps} / \tanh 1 \leq \sqrt\eps / \tanh 1 \leq 100\sqrt{\eps}. \]
    \end{itemize}

    To continue, let $s$ be the hyperbolic projection of $p$ onto $qq'$.
    This makes $p'q'sp$ a \emph{Lambert quadrilateral}, a quadrilateral with three right angles and an acute angle at $p$ that we will call~$\beta$.
    This gives us the following formula \cite[Theorem~32.21]{MR428183}:
    \[ \cos\beta = \sinh|p'q'| \sinh|q's|. \]
    Note here that $|p'q'| \leq 1 + 2\sqrt\eps \leq 3$ and $|q's| \leq \sqrt\eps$, thus $\cos\beta \leq \sinh 3 \sinh 1 \cdot \sqrt\eps \leq 12\sqrt\eps$.
    Both for $\beta$ and for~$\delta$ we can use that $\cos t \geq 1 - 2t / \pi$ to get $\beta \geq \pi/2 - 6\pi\sqrt\eps$ and $\delta \geq \pi/2 - 50\pi \cdot \sqrt\eps$.
    Let $\gamma$ be the hyperbolic angle at $p$ in triangle $psq$; now $\gamma = \pi - \beta - \delta \leq 56\pi \cdot \sqrt{\eps}$.
    By hyperbolic trigonometry, we get
    \[ \tanh|pq| = \tanh|ps| / \cos\gamma. \]
    Here, note that $1 / \cos t \leq 1 + t^2$ for $t \in [0,1]$ to get
    \[ \tanh|pq| \leq (1 + 3136\pi^2 \eps) \tanh|ps|. \]
    By $p'q'sp$ being a Lambert quadrilateral, we also get
    \[ \tanh|ps| = \cosh|q's| \tanh|p'q'|, \]
    where we use $\cosh t \leq 1 + t^2$ for $t \in [0,2]$ and recall that $|q's| \leq |qq'| \leq \sqrt{\eps}$ to get
    \[ \tanh|ps| \leq (1 + \eps) \tanh|p'q'|. \]
    Thus,
    \begin{align*}
         \tanh|pq| \leq (1 + 3136\pi^2 \eps)(1 + \eps) \tanh|p'q'|
         &\leq \left(1 + (3136\pi^2 +1)\eps + 3136\pi^2\eps^2 \right) \tanh|p'q'| \\
         &< \left(1 + \eps \cdot 10^5 \right) \tanh|p'q'|.
    \end{align*}
    If we set $\eps' = \eps / 4 \cdot 10^5$, then Lemma~\ref{lem:tanhapprox} now implies $|pq| \leq (1+\eps') |p'q'|$ as required.
\end{proof}

\paragraph*{Construction.}
Given a point $o$, we can now set the origin of the Klein disk to $o$
and use the Euclidean non-crossing Steiner spanner construction from \cite{planesteinernew} on a set of points
within distance $\frac{1}{4}$ of $o$ to get a graph $G^o$ (where we discard Steiner points at distance further than $\frac{1}{2}$ from $o$ as they will never be used).
We will prove that $G^o$ is also a non-crossing Steiner spanner with (up to constants) the same sparsity $\bigO{\eps^{-3/2}}$.
For this, note that $G^o$ is formed by overlaying graphs $G^o_i$ for $1 \leq i \leq k = \bigO{\sqrt{1/\eps}}$, which correspond to $k$ different directions from
$o$.
Each such graph $G^o_i$ is plane and only has a linear number of edges, but to ensure the combined graph $G^o$ is still plane \cite{planesteinernew} adds Steiner points wherever two graphs $G^o_i$ and $G^o_j$ intersect.
The following lemma (analogous to Lemma~14 of \cite{planesteinernew}) bounds the number of intersections even between graphs
with a different origin
(a generalisation that we will need later).

\begin{lemma}\label{lem:coneintersect}
    Let $uv$ be an edge of some graph $G^o_i$.
    Then,
    any other graph $G^{o'}_j$ contains at most $\bigO{1/\sqrt{\eps}}$ edges that both intersect $uv$ and are at least as long as $uv$.
\end{lemma}
\begin{proof}
    Consider the Klein disk with its origin set to $o'$.
    Let $E'$ be the edges of $G^{o'}_j$ that are at least as long as $uv$.
    Since any edge of $G^{o'}_j$ stays within distance $\frac{1}{2}$ from $o'$, by Lemma~\ref{lem:curvelength} their Euclidean length is at least $1/e$ times their hyperbolic length.
    Meanwhile, the Euclidean length of $uv$ is at most its hyperbolic length.
    Thus, we can split $uv$ into $\lceil e \rceil = 3$ sub-edges whose Euclidean length is shorter than the Euclidean length of any edge in $E'$.
    For each sub-edge~$\sigma$, we can now follow the argument of Lemma~14 from \cite{planesteinernew}.
    Note that all arguments are based on the length of $\sigma$ and on $G^{o'}_j$, which appears exactly as in the Euclidean construction due to $o'$ being the origin.
    Thus, $G^{o'}_j$ contains at most $\bigO{1/\sqrt{\eps}}$ edges that both intersect $\sigma$ and have Euclidean length at least as long as $\sigma$.
    Multiplying by $3$ now gives a bound on the number of edges from $G^{o'}_j$ that both intersect $uv$ and have hyperbolic length at least as long as $uv$.
\end{proof}

Putting Lemma~\ref{lem:conerestrict} and Lemma~\ref{lem:coneintersect} together now gives the following result.

\thmconstdiam*
\begin{proof}
    We set $o$ to be an arbitrary point from $P$ and then construct $G^o$ with $\eps' = \eps \cdot 10^{-5}$.
    First, take $a,b \in P$ arbitrarily so that we can show that $G^o$ is a Steiner $(1+\eps)$-spanner.
    We use that $G^o$ is a Euclidean cone-restricted Steiner $(1+\eps)$-spanner, meaning $G^o$ contains a path $a=p_0, p_1, \dots, p_m=b$ such that we have Euclidean angle $\angle(ab, v_{i-1}v_i) \leq \sqrt{\eps'}$ for all $i = 1, \dots, m$.
    Recall that $a$ and $b$ both have distance at most $\frac{1}{4}$ from $o$, making the maximum Euclidean distance of any point $p_i$ from $ab$ at most $\frac{1}{2}$ by Lemma~\ref{lem:curvelength}.
    Now, also, the maximum Euclidean distance for any point $p_i$ to $ab$ is $\frac{1}{4} \sin\sqrt{\eps'} \leq \frac{1}{4} \sqrt{\eps'}$, meaning that the hyperbolic distance is at most $e/4 \cdot \sqrt{\eps'} < \sqrt{\eps'}$ by Lemma~\ref{lem:curvelength} (using that $p_i$ lies within distance $\frac{1}{2}$ of $o$).
    Lemma~\ref{lem:conerestrict} now proves that $G^o$ is indeed a Steiner $(1+\eps)$-spanner.
    
    Next, for each $G_i^o$, each of its $\bigO{n}$ edges intersects at most $\bigO{1/\eps}$ longer edges by Lemma~\ref{lem:coneintersect}.
    Thus, when overlaying $\bigO{1/\sqrt{\eps}}$ graphs $G_i^o$ and introducing Steiner points on the intersections, each of the in total $\bigO{n / \sqrt{\eps}}$ edges is replaced by at most $\bigO{1/\eps}$ edges, giving $\bigO{n / \eps^{3/2}}$ edges in total.
    Since we use the same construction as in \cite{planesteinernew}, we match their running time.
    The statement for multiple plane Steiner spanners follows by the same reasoning.
    If we want a non-plane Steiner spanner, then we simply take the union of $\bigO{1/\sqrt{\eps}}$ graphs $G_i^o$ with $\bigO{n}$ edges each.
    The plane Steiner spanner did not rely on paths that use edges from multiple graphs $G_i^o$, so this is still a Steiner spanner by the same arguments.
\end{proof}

\section{Conclusion}
We have shown that sparse crossing and non-crossing Steiner spanners can be constructed efficiently on hyperbolic surfaces, and that sparse non-crossing Steiner spanners exist on the hyperbolic plane. This, however, leaves open other quality measures, such as lightness, which we did not try to optimize: the lightness obtained with black-box techniques~\cite{LeS23} does not yet reach the lower bound of $\bigOm{\frac1\eps}$~\cite{BhoreT22}, so this remains an open question.

There are several other fascinating questions that remain to be explored.
\begin{itemize}
    \item On hyperbolic surfaces, the thick-thin decomposition lets us find the systole (the shortest non-contractible loop) when its length is at most $2\mu_M$, but otherwise efficiently finding it remains an open question. Is there a polynomial algorithm?
    \item What is the fine-grained complexity of computing the shortest path between a pair of points on a hyperbolic surface; is it polynomial or exponential in $g$? What about the complexity of the thick-thin decomposition, or of detecting if there is a thin part in a surface?
    \item In our Steiner spanners on surfaces the dependence on $n$ and $\eps$ is optimal, but no lower bound exists for the dependence on $g$, and for the non-crossing Steiner spanner in the hyperbolic plane a construction that matches the Euclidean bounds seems reasonable to expect. Can those be achieved?
    \item Can we generalize some of these ideas to general smooth surfaces, as well as to polyhedral surfaces (sometimes called \emph{portalgons})? Can we construct sparse and light spanners on any surface efficiently?
\end{itemize}

\section*{Acknowledgements}

ChatGPT 5.5 was used to search for related literature and to lightly edit the introduction to improve its flow in a collaborative fashion.

\bibliography{bibliography}
\clearpage
\appendix

\end{document}